\documentclass[12pt]{article}

\usepackage{amsmath}
\usepackage{amssymb}
\usepackage{amsthm}
\usepackage{fullpage}
\usepackage{ccfonts}
\usepackage[T1]{fontenc}

\newcommand{\dunw}{d_{\mathrm{unw}}}
\newcommand{\eps}{\varepsilon}

\newtheorem{definition}{Definition}
\newtheorem{theorem}{Theorem}
\newtheorem{corollary}{Corollary}

\title{On $\eps$-Nets, Distance Oracles, \\ and Metric Embeddings}
\author{Ilya Razenshteyn \footnote{Mathematics Department, Moscow State University, e-mail: \texttt{ilyaraz@gmail.com}}}
\date{}

\begin{document}
    \maketitle
    \begin{abstract}
        We give two new applications of an observation from~\cite{ADFGW11}.
        The first is an almost linear sized constant time data structure for reporting very
        large distances in undirected graphs.
        The second is a generic transformation of results about $\ell_1$-embeddability of metrics to a setting,
        where we are interested in preservation of large distances only.
    \end{abstract}
    \section{Introduction}

    Let $G$ be a weighted undirected graph with unique shortest paths (with non-negative weights).
    In~\cite{ADFGW11} the following observation is heavily used: the set of all shortest paths of $G$ (by shortest path we mean its set of vertices)
    has VC-dimension~\cite{VC71} at most two.
    We give two new applications of this fact.

    \subsection{Distance oracles}

    The first application is about distance oracles.
    Let $G = (V, E)$ be an undirected unweighted graph with $n$ vertices.
    For vertices $v_1, v_2 \in V$ let $d(v_1, v_2)$ be the distance between $v_1$ and $v_2$.
    Let $\eps, \delta$ be two fixed positive constants.
    In Section~\ref{data_structure}
    we build a data structure of size $O(n^{1 + \delta})$ that given two vertices $v_1, v_2 \in V$ in constant time reports the following:
    \begin{itemize}
        \item ``$\perp$'', if $d(v_1, v_2) < \eps n$;
        \item $d(v_1, v_2)$, if $d(v_1, v_2) \geq \eps n$.
    \end{itemize}

    \subsection{Metric embeddings}

    The second application is about metric embeddings.
    Let $G = (V, E, w)$ be an undirected weighted graph with $n$ vertices.
    For vertices $v_1, v_2 \in V$ let $d(v_1, v_2)$ be the distance between $v_1$ and $v_2$ with respect to $w$,
    and $\dunw(v_1, v_2)$ be the distance with respect to unit weights.

    Suppose we want to approximate metric $d$ with $\ell_1$-norm. We are looking for mappings from $V$ to $\ell_1$ with small distortion. 
    \begin{definition}
        Let us say that a mapping $f \colon V \to (\mathbb{R}^k, \ell_1)$ has distortion $D$ if for every $v_1, v_2 \in V$
        $$
            d(v_1, v_2) \leq \|f(v_1) - f(v_2)\|_1 \leq D \cdot d(v_1, v_2).
        $$
    \end{definition}

    There are many results about $\ell_1$-embeddability. Let us state some of them.

    \begin{definition}
        Doubling dimension of $d$ is the minimum integer $k$ such that every subset $V' \subseteq V$ of diameter $\Delta$
        can be covered with $2^k$ subsets of diameter $\Delta / 2$.
    \end{definition}

    \begin{definition}
        Let us say that $d$ is of negative type if $\sqrt{d}$ is isometrically embeddable into $\ell_2$.
    \end{definition}

    \begin{theorem}
        \label{l1_embeddings}
        There are the following upper-bounds on $D$:
        \begin{itemize}
            \item \cite{B85} For any $d$ one can take $D = O(\log n)$;
            \item \cite{R99} If $G$ is $H$-minor free, then $D = O(\sqrt{\log n})$; 
            \item \cite{GKL03} If $d$ has bounded doubling dimension, then $D = O(\sqrt{\log n})$;
            \item \cite{ALN05} If $d$ is of negative type, then $D = O(\sqrt{\log n} \log \log n)$. 
        \end{itemize}
    \end{theorem}

    But what if we are interested in preserving $d(v_1, v_2)$ only if $\dunw(v_1, v_2) \geq \eps n$?
    In Section~\ref{metric_embeddings} we state and prove a generic black box transformation
    that allows us to replace all occurences of $n$ in Theorem~\ref{l1_embeddings} with $1 / \eps$ for this case.

    In~\cite{ABCDGKNS05} somewhat stronger result about arbitrary, $H$-minor free, and bounded doubling dimension metrics is proved, but since our transformation
    is black box, it automatically holds for negative type metrics.
    \section{VC-dimension and $\eps$-nets}

    Let $X$ be a finite set. Let $R \subseteq 2^X$ be a system of subsets of $X$.
    The following definitions were given in~\cite{VC71}.

    \begin{definition}
        Say that $U \subseteq X$ is \emph{shattered} by $R$ if for every $U' \subseteq U$ there exists $U'' \in R$ such that $U'' \cap U = U'$.
    \end{definition}
    \begin{definition}
        \emph{VC-dimension} of the pair $(X, R)$ is the size of the largest subset of $X$, which is shattered by $R$.
    \end{definition}

    Consider the following set system: let $G = (V, E, w)$ be a weighted undirected graph with unique shortest paths.
    Let $X = V$, and $P \in R$ iff $P$ is a shortest path in $G$. The following theorem was proved in~\cite{ADFGW11}. We reprove it here for completeness.

    \begin{theorem}
        \label{vc_shortest}
        VC-dimension of $(X, R)$ is at most two.
    \end{theorem}
    \begin{proof}
        Let $\{v_1, v_2, v_3\} \subseteq V$. Suppose it is shattered by shortest paths in $G$.
        Then there exists a shortest path $P$ that contains $v_1$, $v_2$, and $v_3$ simultaneously.
        W.l.o.g. we can assume that $v_1$, $v_2$, and $v_3$ appear in $P$ exactly in the same order.
        But then, since shortest paths in $G$ are unique, there is no shortest path that contains $v_1$ and $v_3$, but does not contain $v_2$.
    \end{proof}

    If VC-dimension of $(X, R)$ is small, then there exist small the so-called $\eps$-nets.

    \begin{definition}
        Let $\mu$ be a probabilistic measure over $X$, and $\eps$ be a fixed positive constant.
        Say that $N \subseteq X$ is an $\eps$-net for $(X, \mu, R)$ if for every $U \in R$ such that $\mu(U) \geq \eps$ the intersection of $N$ and $U$
        is not empty.
    \end{definition}

    The following theorem was proved in~\cite{HW86}.

    \begin{theorem}
        \label{eps_nets}
        If VC-dimension of $(X, R)$ is at most $d$, then for every $\eps > 0$ and every $\mu$ there exists an $\eps$-net of $(X, \mu, R)$ of size 
        $O(d \cdot \log(1 / \eps) / \eps)$.
    \end{theorem}

    By combining Theorem~\ref{vc_shortest} and Theorem~\ref{eps_nets} we obtain the main ingredient of two our results.

    \begin{theorem}
        \label{hitting_set}
        Let $G = (V, E, w)$ be a weighted undirected graph with $n$ vertices.
        Let $d$ be a shortest-path metric on $V$ with respect to $w$, and $\dunw$~--- with repect to unit weights.

        Then for every $\eps > 0$ there exists a subset $U \subseteq V$ of size $O(\log(1 / \eps) / \eps)$
        such that for every $v_1, v_2 \in V$ such that $\dunw(v_1, v_2) \geq \eps n$
        there exists $u \in U$ such that $d(v_1, v_2) = d(v_1, u) + d(u, v_2)$.
    \end{theorem}
    \begin{proof}
        By perturbing weights, we can assume that shortest paths in $G$ are unique. Then we just apply Theorem~\ref{eps_nets} for uniform $\mu$
        and $d = 2$ (by Theorem~\ref{vc_shortest}).
    \end{proof}

    It follows from~\cite{A10} that one can not improve the bound on the size of $U$ to $O(1 / \eps)$, but if $G$ is $H$-minor
    free, then it follows from~\cite{KPR93} that there exists a desired $U$ of size $O(1 / \eps)$.
    \section{Distance oracles}
    \label{data_structure}

    Let $G = (V, E)$ be an undirected unweighted graph. In this section we show how to build a small and effective distance oracle 
    that reports large distances in $G$.

    \begin{theorem}
        \label{simple_oracle}
        For every $\eps > 0$
        there exists an oracle of size $O(n)$ that given two vertices $v_1, v_2 \in V$ reports in constant time
        a number $\tau$ such that if $d(v_1, v_2) \geq \eps n$, then
        $\tau = d(v_1, v_2)$.
    \end{theorem}
    \begin{proof}
        Let $U \subseteq V$ be a set from Theorem~\ref{hitting_set}.
        Let us store for every vertex $v$ distances to all elements of $U$.
        Then if we are given a query $(v_1, v_2)$, we report $\tau := \min_{u \in U} d(v_1, u) + d(u, v_2)$. By definition of $U$, if $d(v_1, v_2) \geq \eps n$,
        then $\tau = d(v_1, v_2)$.
    \end{proof}

    Since we want our oracle to distinguish the cases $d(v_1, v_2) < \eps n$ and $d(v_1, v_2) \geq \eps n$, we combine Theorem~\ref{simple_oracle}
    with Thorup-Zwick oracle~\cite{TZ05}.

    \begin{theorem}[\cite{TZ05}]
        \label{thorup_zwick}
        For every positive integer $k$ there exists a data structure of size $O(n^{1 + 1/k})$ that given two vertices $v_1, v_2 \in V$ reports
        in constant time a number $D$ such that $d(v_1, v_2) \leq D \leq (2k - 1) \cdot d(v_1, v_2)$.
    \end{theorem}

    Now we show how one can combine Theorem~\ref{simple_oracle} and Theorem~\ref{thorup_zwick} and obtain the desired data structure.

    \begin{theorem}
        For every $\eps, \delta > 0$
        there exists an oracle of size $O(n^{1 + \delta})$ that given two vertices $v_1, v_2 \in V$ reports in constant time
        \begin{itemize}
            \item ``$\perp$'', if $d(v_1, v_2) < \eps n$;
            \item $d(v_1, v_2)$, if $d(v_1, v_2) \geq \eps n$.
        \end{itemize}
    \end{theorem}
    \begin{proof}
        First, we build a Thorup-Zwick oracle for $k = \lceil 1 / \delta \rceil$.
        Using it we can dismiss pairs $(v_1, v_2)$ with $d(v_1, v_2) < \eps n / (2k - 1)$.
        Now by using an oracle from Theorem~\ref{simple_oracle} with $\eps' = \eps / (2k - 1)$ we can determine the exact value of $d(v_1, v_2)$.
    \end{proof}

    \section{Metric embeddings}
    \label{metric_embeddings}

    In this section we state and prove a generic transformation of $\ell_1$-embeddability results of metric spaces to a setting, where we are concerned
    only about large distances.

    Let $G = (V, E, w)$ be a weighted undirected graph with $n$ vertices.
    Let $d$ and $\dunw$ be weighted and unweighted metric on $V$, respectively.
    Let $\eps > 0$ be some parameter.

    \begin{definition}
        Say that a mapping $f \colon V \to (\mathbb{R}^k, \ell_1)$ has $\eps$-distortion $D$, if 
        \begin{itemize}
            \item For every $v_1, v_2 \in V$
            $$
                \|f(v_1) - f(v_2)\|_1 \geq d(v_1, v_2); 
            $$
            \item For every $v_1, v_2 \in V$ such that $\dunw(v_1, v_2) \geq \eps n$
            $$
                \|f(v_1) - f(v_2)\|_1 \leq D \cdot d(v_1, v_2).
            $$
        \end{itemize}
    \end{definition}

    \begin{theorem}
        \label{black_box}
        Suppose that for every $k \leq n$ and for every subset $U \subseteq V$ of size $k$ there exists an embedding 
        $g_U \colon U \to (\mathbb{R}^{t(k)}, \ell_1)$ with distortion $D(k)$.

        Then for every $\eps > 0$ there exists an embedding $f \colon V \to (\mathbb{R}^{t(O(\log(1 / \eps) / \eps)) + O(\log n)}, \ell_1)$ with $\eps$-distortion
        $O(D(O(\log(1 / \eps) / \eps)))$.
    \end{theorem}
    \begin{proof}
        Let $U \subseteq V$ be a set of size $k := O(\log(1 / \eps) / \eps)$ from Theorem~\ref{hitting_set}.
        There exists an embedding $g \colon U \to (\mathbb{R}^{t(k)}, \ell_1)$ with distortion $D(k)$.

        For a vertex $v \in V$ let $p(v) \in V$ be the closest to $v$ vertex from $U$ (with respect to $d$).
        Consider the metric $d'(v_1, v_2) := d(v_1, p(v_1)) + d(v_2, p(v_2))$. It is well known that such metrics
        are embeddable into $(\mathbb{R}^{O(\log n)}, \ell_1)$ with distortion $O(1)$.
        Let $h \colon V \to (\mathbb{R}^{O(\log n)}, \ell_1)$ be a corresponding embedding.

        Let us prove that a mapping $f \colon v \mapsto g(p(v)) \oplus h(v)$
        has $\eps$-distortion $O(D(k))$.

        First, for every $v_1, v_2 \in V$
        \begin{eqnarray*}
            \|f(v_1) - f(v_2)\|_1 = \|g(p(v_1)) - g(p(v_2))\|_1 + \|h(v_1) - h(v_2)\|_1 \geq \\ \geq
            d(p(v_1), p(v_2)) + d'(v_1, v_2) = d(v_1, p(v_1)) + d(p(v_1), p(v_2)) + d(p(v_2), v_2) \geq \\ \geq
            d(v_1, v_2).
        \end{eqnarray*}

        Second, suppose $\dunw(v_1, v_2) \geq \eps n$. Then, by definition of $U$, there exists $u \in U$
        such that $d(v_1, v_2) = d(v_1, u) + d(u, v_2)$.
        Thus,
        \begin{eqnarray*}
            \|f(v_1) - f(v_2)\|_1 = \|g(p(v_1)) - g(p(v_2))\|_1 + \|h(v_1) - h(v_2)\|_1 \leq \\ \leq
            D(k) \cdot d(p(v_1), p(v_2)) + O(1) \cdot d'(v_1, v_2) \leq \\ \leq
            O(D(k)) \cdot (d(p(v_1), p(v_2)) + d(v_1, p(v_1)) + d(v_2, p(v_2))) \leq \\ \leq
            O(D(k)) \cdot (d(v_1, v_2) + d(v_1, p(v_1)) + d(v_2, p(v_2))) \leq \\ \leq
            O(D(k)) \cdot (d(v_1, v_2) + d(v_1, u) + d(v_2, u)) = 
            O(D(k)) \cdot d(v_1, v_2).
        \end{eqnarray*}
    \end{proof}

    By combining Theorem~\ref{black_box} and Theorem~\ref{l1_embeddings} we obtain the following corollary.
    \begin{corollary}
        There are the following upper-bounds on $\eps$-distortion $D$:
        \begin{itemize}
            \item For any $d$ one can take $D = O(\log (1 / \eps))$;
            \item If $G$ is $H$-minor free, then $D = O(\sqrt{\log (1 / \eps)})$; 
            \item If $d$ has bounded doubling dimension, then $D = O(\sqrt{\log (1 / \eps)})$;
            \item If $d$ is of negative type, then $D = O(\sqrt{\log (1 / \eps)} \log \log (1 / \eps))$. 
        \end{itemize}
    \end{corollary}

    \bibliographystyle{alpha}
    \bibliography{../bibtex/ir}
\end{document}